\documentclass[a4paper,11pt,reqno,
final
]{amsart}	

\usepackage{amsmath}
\usepackage{amssymb}
\usepackage{amsthm}
\usepackage{graphicx}% Include figure files
\usepackage{dcolumn}% Align table columns on decimal point
\usepackage{bm}% bold math
\usepackage{psfrag}      % see /usr/share/texmf/doc/latex/psfrag/pfgguide.ps
\usepackage{subfig}
\usepackage{pstricks}
\usepackage{setspace}  %to get double space% or use the graphicx package for more complicated commands
\usepackage[latin1]{inputenc}
\usepackage[english]{babel}

\newtheorem{theorem}{Theorem}[section]
\newtheorem{lemma}[theorem]{Lemma}
\newtheorem{proposition}[theorem]{Proposition}
\newtheorem{corollary}[theorem]{Corollary}
\newtheorem{remark}[theorem]{Remark}

\def\R{{\mathbb R}}

\def\F{{\mathbb F}}

\def\Hcal{\mathcal{H}}

\def\Acal{{\mathcal A}}

\def\Fcal{{\mathcal F}}

\def\Hcal{{\mathcal H}}

\def\rmd{{\mathrm d}}
\def\rme{{\mathrm e}}

\title[Optimal Life Insurance Purchase]{Optimal Life Insurance Purchase, Consumption and Investment on a financial market with multi-dimensional diffusive terms}

\author[I. Duarte]{I. Duarte}
\address[I. Duarte]{Centro de Matemática da Universidade do Minho, Braga, Portugal}
\email{isabelduarte@math.uminho.pt}

\author[D. Pinheiro]{D. Pinheiro}
\address[D. Pinheiro]{CEMAPRE, ISEG, Technical University of Lisbon, Lisbon, Portugal}
\email{dpinheiro@iseg.utl.pt}

\author[A. A. Pinto]{A. A. Pinto}
\address[A. A. Pinto]{LIAAD-INESC Porto LA and Dep of Mathematics, Faculty of Science, University of Porto, Portugal}
\email{aapinto@fc.up.pt}

\author[S. R. Pliska]{S. R. Pliska}
\address[S. R. Pliska]{Dept. of Finance, University of Illinois at Chicago, Chicago, IL 60607, USA}
\email{srpliska@uic.edu}

\begin{document}

\begin{abstract}
We introduce an extension to Merton's famous continuous time model of optimal consumption and investment, in the spirit of previous works by Pliska and Ye, to allow for a wage earner to have a random lifetime and to use a portion of the income to purchase life insurance in order to provide for his estate, while investing his savings in a financial market comprised of one risk-free security and an arbitrary number of risky securities driven by multi-dimensional Brownian motion. We then provide a detailed analysis of the optimal consumption, investment, and insurance purchase strategies for the wage earner whose goal is to maximize the expected utility obtained from his family consumption, from the size of the estate in the event of premature death, and from the size of the estate at the time of retirement. We use dynamic programming methods to obtain explicit solutions for the case of discounted constant relative risk aversion utility functions and describe new analytical results which are presented together with the corresponding economic interpretations.
\end{abstract}

\maketitle

\section{Introduction}

We consider the problem faced by a wage earner having to make decisions continuously about three strategies: consumption, investment and life insurance purchase during a given interval of time  $[0,\min\{T,\tau\}]$, where $T$ is a fixed point in the future that we will consider to be the retirement time of the wage earner and $\tau$ is a random variable representing the wage earner's time of death. We assume that the wage earner receives his income at a continuous rate $i(t)$ and that this income is terminated when the wage earner dies or retires, whichever happens first. One of our key assumptions is that the wage earner's lifetime $\tau$ is a random variable and, therefore, the wage earner needs to buy life insurance to protect his family for the eventuality of premature death. The life insurance depends on a insurance premium payment rate $p(t)$ such that if the insured pays $p(t)\cdot\delta  t$ and dies during the ensuing short time interval of length $\delta t$ then the insurance company will pay $p(t)/\eta(t)$ dollars to the insured's estate, where $\eta(t)$ is an amount set in advance by the insurance company. Hence this is like term insurance with an infinitesimal term. We also assume that the wage earner wants to maximize the satisfaction obtained from a consumption process with rate $c(t)$. In addition to consumption and purchase of a life insurance policy, we assume that the wage earner invests the full amount of his savings in a financial market consisting of one risk-free security and a fixed number $N\ge 1$ of risky securities with diffusive terms driven by $M$-dimensional Brownian motion.

The wage earner is then faced with the problem of finding strategies that maximize the utility of (i) his family consumption for all $t \leq \min\{T,\tau\}$; (ii) his wealth at retirement date $T$ if he lives that long; and (iii) the value of his estate in the event of premature death. Various quantitative models have been proposed to model and analyze this kind of problem, at least problems having at least one of these three objectives. This literature is perhaps highlighted by Yarri \cite{Yarri} who considered the problem of optimal financial planning decisions for an individual with an uncertain lifetime as well as by Merton \cite{Merton_1,Merton_2} who emphasized optimal consumption and investment decisions but did not consider life insurance. These two approaches were combined by Richard, whose impressive paper \cite{Richard} uses sophisticated methods at an early date for the analysis of a life-cycle life insurance and consumption-investment problem in a continuous time model. Later, Pliska and Ye \cite{Pliska_Ye_1,Pliska_Ye_2} introduced a continuous-time model that combined the more realistic features of all those in the existing literature and extended the model proposed previously by Richard, the main difference being the choice of the boundary condition, leading to somewhat different economic interpretations of the underlying problem. More precisely, while Richard assumed that the lifetime of the wage earner is limited by some fixed number, the model introduced by Pliska and Ye had the feature that the duration of life is a random variable which takes values in the interval $]0,\infty[$ and is independent of the stochastic process defining the underlying financial market. Moreover, Pliska and Ye made the following refinements to the theory: (i) the planning horizon $T$ is now seen as the moment when the wage earner retires, contrary to Richard's interpretation as a finite upper bound on the lifetime; and (ii) the utility of the wage earner's wealth at the planning horizon $T$ is taken into account as well as the utility of lifetime consumption and the utility of the bequest in the event of premature death. Blanchet-Scalliet et al. paper \cite{Scalliet_Karoui_Jeanblanc_Martellini}  deals with optimal portfolio selection with an uncertainty exit time for a suitable extension of the familiar optimal consumption investment problem of Merton, without considering any kind of life insurance purchase.

Whereas Pliska and Ye's financial market involved only one security that was risky, in the present paper we study the extension where there is an arbitrary (but finite) number of risky securities. The existence of these extra risky securities gives greater freedom for the wage earner to manage the interaction between his life insurance policies and the portfolio containing his savings invested in the financial market. Some examples of these interactions are described below.

Following Pliska and Ye, we use the model of uncertain life found in reliability theory, commonly used for industrial life-testing and actuarial science, to model the uncertain time of death for the wage earner. This enables us to replace Richard's assumption that lifetimes are bounded with the assumption that lifetimes take values in the interval $]0,\infty[$. We then set up the wage earner's objective functional depending on a random horizon $\min\{T,\tau\}$ and transform it to an equivalent problem having a fixed planning horizon, that is,  the wage earner who faces unpredictable death acts as if he will live until some time $T$, but with a subjective rate of time preferences equal to his ``force of mortality'' for his consumption and terminal wealth. This transformation to a fixed planning horizon enables us to state the dynamic programming principle and derive an associated Hamilton-Jacobi-Bellman (HJB) equation. We use the HJB equation to derive the optimal feedback control, that is, the optimal insurance, portfolio and consumption strategies. Furthermore, we obtain explicit solutions for the family of discounted Constant Relative Risk Aversion (CRRA) utilities and examine the economic implications of such solutions.

In the case of discounted CRRA utilities our results generalize those obtained previously by Pliska and Ye. For instance, we obtain: (i) an economically reasonable description for the optimal expenditure for insurance as a decreasing function of the wage earner's overall wealth and a unimodal function of age, reaching a maximum at an intermediate age; (ii) a more controversial conclusion that possibly an optimal solution calls for the wage earner to sell a life insurance policy on his own life late in his career. Nonetheless, the extra risky securities in our model introduce novel features to the wage earner's portfolio and insurance management interaction such as: (i) a young wage earner with small wealth has an optimal portfolio with larger values of volatility and higher expected returns, with the possibility of having short positions in lower yielding securities; and (ii) a wage earner who can buy life insurance policies will choose a more conservative portfolio than a wage earner who is without the opportunity to buy life insurance, the distinction being clearer for young wage earners with low wealth.

This paper is organized as follows. In section \ref{formulation} we describe the problem we propose to address. Namely, we introduce the underlying financial and insurance markets, as well as the problem formulation from the point of view of optimal control. In section \ref{SDP} we see how to use the dynamic programming principle to reduce the optimal control of section \ref{formulation} to one with a fixed planning horizon and then derive an associated HJB equation. We devote section \ref{CRRA} to the case of discounted CRRA utilities. We conclude in section \ref{conclusion}.

\section{Problem formulation}\label{formulation}

Throughout this section, we define the setting in which the wage earner has to make his decisions regarding consumption, investment and life insurance purchase. Namely, we introduce the specifications regarding the financial and insurance markets available to the wage earner. We start by the financial market description, followed by the insurance market and conclude with the definition of a wealth process for the wage earner.

\subsection{The financial market model}

We consider a financial market consisting of one risk-free asset and several risky-assets. Their respective prices $(S_0(t))_{0\le t\le T}$ and $(S_n(t))_{0\le t\le T}$ for $n = 1,...,N$ evolve according to the equations:
\begin{eqnarray*}\label{eq1}
\rmd S_0(t) &=& r(t)S_0(t)\rmd t \ , \hspace{5.2 cm}  S_0(0)=s_0 \ , \nonumber  \\
\rmd S_n(t) &=& \mu_n(t)S_n(t)\rmd t + S_n(t)\sum_{m = 1}^M\sigma_{nm}(t)\rmd W_m(t)\ , \quad S_n(0)=s_n>0 \ ,
\end{eqnarray*}
where $W(t) = (W_1(t),\ldots,W_M(t))^T$ is a standard $M$-dimensional Brownian motion on a probability space $(\Omega,\Fcal,P)$, $r(t)$ is the riskless interest rate, $\mu(t)=(\mu_1(t),\ldots,\mu_N(t))\in\R^N$ is the vector of the risky-assets appreciation rates and $\sigma(t)=(\sigma_{nm}(t))_{1\le n\le N, 1\le m\le M}$ is the matrix of risky-assets volatilities.

We assume that the coefficients $r(t)$, $\mu(t)$ and $\sigma(t)$ are deterministic continuous functions on the interval $[0,T]$. We also assume that the interest rate $r(t)$ is positive for all $t\in[0,T]$ and the matrix $\sigma(t)$ is such that $\sigma \sigma^T$ is nonsingular for Lebesgue almost all $t \in [0,T]$ and
 satisfies the following integrability condition
\begin{equation*}
\sum_{n=1}^N \sum_{m=1}^M \int_0^T \sigma_{nm}^2(t) \rmd t < \infty  \ .
\end{equation*}
Furthermore, we suppose that there exists an $(\Fcal_t)_{0\le t\le T}$-progressively measurable process $\pi(t)\in\R^M$, called the market price of risk, such that for Lebesgue-almost-every $t\in [0, T]$ the risk premium
\begin{equation}\label{rp}
\alpha(t) = (\mu_1(t)-r(t),\ldots,\mu_N(t)-r(t)) \in \R^N
\end{equation}
is related to $\pi(t)$ by the equation
\begin{equation*}
\alpha(t) = \sigma(t)\pi(t)  \qquad \text{a.s. }
\end{equation*}
and is such that the following two conditions hold
\begin{eqnarray*}
&&\int_0^T\left\|\pi(t)\right\|^2<\infty \qquad \text{a.s.}  \\
&&E\left[\exp\left( -\int_0^T \pi(s) \rmd W(s) - \frac{1}{2} \int_0^T \left\|\pi(s)\right\|^2\rmd s \right)\right] = 1  \ .
\end{eqnarray*}
The existence of such a process $\pi(t)$ ensures the absence of arbitrage opportunities in the financial market defined above. Note also that the conditions on the matrix $\sigma$ above do not imply market completeness. See \cite{Karatzas_Shreve_2} for further details on market viability and completeness.

Moreover, throughout the paper we will assume that $(\Omega,\Fcal,P)$ is a filtered probability space and that its filtration $\F=\{\Fcal_t, t\in[0,T]\}$ is the $P$-augmentation of the filtration generated by the Brownian motion $W(t)$, $\sigma\{W(s)  , s\le t\}$ for $t \ge 0$. Each sub-$\sigma$-algebra $\Fcal_t$ represents the information available to any given agent observing the financial market until time $t$.

\subsection{The life insurance market model}

We assume that the wage earner is alive at time $t = 0$ and that his lifetime is a non-negative random variable $\tau$ defined
on the probability space $(\Omega, \Fcal, P)$. Furthermore, we assume that the random variable $\tau$ is independent of the filtration $\mathbb{F}$ and has a distribution function $F:[0,\infty)\rightarrow[0,1]$ with density $f:[0,\infty)\rightarrow\R^+$ so that
\begin{equation*} \label{eq3}
F(t) = \int_0^t f(s)\;\rmd s \ .
\end{equation*}
We define the \emph{survivor function} $\overline{F}:[0,\infty)\rightarrow[0,1]$ as the probability for the wage earner to survive at least until time $t$, i.e.
\begin{equation*} \label{eq8}
\overline{F}(t) = P(\tau \geq t) = 1 - F(t) \ .
\end{equation*}
We shall make use of the \emph{hazard function}, the conditional, instantaneous death rate for the wage earner surviving to time $t$ , that is
\begin{equation} \label{eq9}
\lambda(t) = \underset{\delta t \rightarrow 0}{\textrm{lim}}~\frac{P(t \leq \tau < t + \delta t~|~\tau \geq t)}{\delta t} = \frac{f(t)}{\overline{F}(t)} \ .
\end{equation}
Throughout the paper, we will suppose that the hazard function $\lambda:[0,\infty)\rightarrow\mathbb{R}^+$ is a continuous and deterministic function such that
\begin{equation*}
\int_0^\infty\lambda(t)\;\rmd t = \infty \ .
\end{equation*}
These two concepts introduced above are standard in the context of reliability theory and actuarial science. In our case, such concepts enable us to consider an optimal control problem with a stochastic planning horizon and restate it as one with a fixed horizon.

Due the uncertainty concerning his lifetime, the wage earner buys life insurance to protect his family for the eventuality of premature death.  The life insurance is available continuously and the wage earner buys it by paying a \emph{insurance premium payment rate} $p(t)$ to the insurance company. The insurance contract is like term insurance, with an infinitesimally small term. If the wage earner dies at time $\tau<T$ while buying insurance at the rate $p(t)$, the insurance company pays an amount $p(\tau)/\eta(\tau)$ to his estate, where $ \eta:[0,T]\rightarrow\mathbb{R}^+$ is a continuous and deterministic function which we call the \emph{insurance premium-payout ratio} and is regarded as fixed by the insurance company. The contract ends when the wage earner dies or achieves retirement age, whichever happens first. Therefore, the wage earner's total legacy to his estate in the event of a premature death at time  $\tau<T$ is given by
\begin{equation*}
Z(\tau) = X(\tau) + \frac{p(\tau)}{\eta(\tau)} \ ,
\end{equation*}
where $X(t)$ denotes the wage earner's savings at time $t$.

\subsection{The wealth process}

We assume that the wage earner receives an income $i(t)$ at a continuous rate during the period $[0, \textrm{min}\{T, \tau\}]$, i.e., the income will be terminated either by his death or his retirement, whichever happens first. Furthermore, we assume that $i:[0,T]\rightarrow\mathbb{R}^+$ is a deterministic Borel-measurable function satisfying the integrability condition
\begin{equation*}
\int_0^T i(t)\; \rmd t < \infty \ .
\end{equation*}

The \emph{consumption process} $(c(t))_{0\le t\le T}$ is a $(\Fcal_t)_{0\le t\le T}$-progressively measurable nonnegative process satisfying the following integrability condition for the investment horizon $T>0$
\begin{equation*}
\int_0^T c(t) \; \rmd t < \infty  \qquad \textrm{a.s.} \ .
\end{equation*}
We assume also that the insurance premium payment rate $(p(t))_{0\le t \le T}$ is a $(\Fcal_t)_{0\le t\le T}$-predictable process, i.e., $p(t)$ is measurable with respect to the smallest $\sigma$-algebra on $\R^+\times\Omega$ such that all left-continuous and adapted processes are measurable. In a intuitive manner, a predictable process can be described as such that its values are ``known'' just in advance of time.

For each $n=0,1,...,N$ and $t\in[0,T]$, let $\theta_n(t)$ denote the fraction of the wage earner's wealth allocated to the asset $S_n$ at time $t$. The \emph{portfolio process} is then given by $\Theta(t)=\left(\theta_0(t),\theta_1(t),\cdots,\theta_N(t)\right)\in\R^{N+1}$, where
\begin{equation}\label{soma_thetas}
\sum_{n=0}^N\theta_n(t) = 1 \ , \qquad 0\le t\le T \ .
\end{equation}
We assume that the portfolio process is $(\Fcal_t)_{0\le t\le T}$-progressively measurable and that, for the fixed investment horizon $T > 0$, we have that
\begin{equation*}
\int_0^T \left\|\Theta(t)\right\|^2\; \rmd t < \infty  \qquad \textrm{a.s.} \ ,
\end{equation*}
where $\left\|\cdot\right\|$ denotes the Euclidean norm in $\R^{N+1}$.

The \emph{wealth process} $X(t)$, $t \in [0, \textrm{min}\{T,\tau\}]$, is then defined by
\begin{equation} \label{eq4}
X(t) = x + \int_0^t \left[i(s)-c(s)-p(s)\right]\;\rmd s+\sum_{n = 0}^N\int_0^t\frac{\theta_n(s)X(s)}{S_n(s)}\;\rmd S_n(s) \ ,
\end{equation}
where $x$ is the wage earner's initial wealth. This last equation can be rewritten in the differential form
\begin{eqnarray} \label{eq5}
\rmd X(t) & = & \left(i(t)-c(t)-p(t)+\bigg(\theta_0(t)r(t)+\sum_{n = 1}^N{\theta_n(t)\mu_n(t)}\bigg)X(t)\right)\rmd t \nonumber \\
& & +\sum_{n=1}^N{\theta_n(t)X(t)\sum_{m = 1}^M{\sigma_{nm}(t)\rmd W_m(t)}} \ ,
\end{eqnarray}
where $0 \le t \leq \min\{\tau,T\}$.

Using relation \eqref{soma_thetas} we can always write $\theta_0(t)$ in terms of $\theta_1(t),\ldots,\theta_N(t)$, so from now on we will define the portfolio process in terms of the \emph{reduced portfolio process} $\theta(t)\in\R^N$ given by 
\begin{equation*}
\theta(t)=\left(\theta_1(t), \theta_2(t), \cdots,\theta_N(t)\right)\in\R^{N} \ .
\end{equation*}

\subsection{The optimal control problem}

The wage earner is faced with the problem of finding strategies that maximize the expected utility obtained from:
\begin{itemize}
\item[(a)] his family consumption for all $t \leq \min\{T,\tau\}$;
\item[(b)] his wealth at retirement date $T$ if he lives that long;
\item[(c)] the value of his estate in the event of premature death.
\end{itemize}

This problem can be formulated by means of optimal control theory: the wage earner's goal is to maximize some cost functional
subject to (i) the (stochastic) dynamics of the state variable, i.e., the dynamics of the wealth process $X(t)$ given by \eqref{eq4}; (ii) constraints on the control variables, i.e., the consumption process $c(t)$, the premium insurance rate $p(t)$ and the portfolio process $\theta(t)$; and (iii) boundary conditions on the state variables.

Let us denote by $\Acal(x)$ the set of all admissible decision strategies, i.e., all admissible choices for the control variables $\nu = (c,p,\theta)\in\R^{N+2}$. The dependence of $\Acal(x)$ on $x$ denotes the restriction imposed on the wealth process by the boundary condition $X(0)=x$. In particular, $\Acal(x)$ must be such that for each $\nu\in\Acal(x)$ the corresponding wealth process satisfies $X(t)\geq 0$ for all $t \leq \min\{T,\tau\}$.

The wage earner's problem can then be restated as follows: find a strategy $\nu=(c,p,\theta)\in\Acal(x)$ which maximizes the expected utility
\begin{eqnarray} \label{eq7}
V(x) = \underset{\nu \in \Acal(x)}{\sup}\; E_{0,x} \left[\int_0^{T \wedge \tau}U(c(s),s) \; \rmd s + B(Z(\tau),\tau)I_{\{\tau \leq T\}}  + W(X(T))I_{\{\tau > T\}} \;   \right] \ ,
\end{eqnarray}
where $T \wedge \tau=\min\{T,\tau\}$, $I_A$ denotes the indicator function of event $A$, $U(c,\cdot)$ is the utility function describing the wage earner's family preferences regarding consumption in the time interval $[0,\min\{T,\tau\}]$, $B(Z,\cdot)$ is the utility function for the size of the wage earners's legacy in case $\tau\le T$, and $W(X)$ is the utility function for the terminal wealth at time $t=T$ in the case $\tau>T$.

We suppose that $U$ and $B$ are strictly concave on their first variable and that $W$ is strictly concave on its sole variable. In section \ref{CRRA} we specialize our analysis to the case where the wage earner's preferences are described by discounted CRRA utility functions.

\section{Stochastic optimal control}\label{SDP}

In this section we use the techniques in \cite{Pliska_Ye_1,Pliska_Ye_2} to restate the stochastic optimal control problem formulated in the preceding section as one with a fixed planning horizon and to derive a dynamic programming principle and the corresponding HJB equation.

\subsection{Dynamic programming principle}

Let us denote by $\Acal(t,x)$ the set of admissible decision strategies $\nu=(c,p,\theta)$ for the dynamics of the wealth process with boundary condition $X(t)=x$. For any $\nu\in\Acal(t,x)$ we define the functional
\begin{eqnarray*}
J(t,x;\nu) =   E_{t,x} \left[\int_t^{T \wedge \tau}U(c(s),s) \; \rmd s + B(Z(\tau),\tau)I_{\{\tau \leq T\}}  + W(X(T))I_{\{\tau > T\}} \; \Bigl|\;\tau>t ,\Fcal_t  \right]
\end{eqnarray*}
and introduce the associated value function:
\begin{equation*}
V(t,x) = \underset{\nu \in \Acal(t,x)}{\sup}\; J(t,x;\nu) \ .
\end{equation*}

Given some initial condition $(t,x) \in[0,T]\times\R$, we have that $\nu^*\in\Acal(t, x)$ is an optimal control if $V(t,x) = J(t, x; \nu^*)$.

The following lemma is the key tool to restating the control problem above as an equivalent one with a fixed planning horizon. See \cite{Ye} for a proof.
\begin{lemma}\label{lemma1}
Suppose that the utility function $U$ is either nonnegative or nonpositive. If the random variable $\tau$ is independent of the filtration $\F$, then
\begin{eqnarray*} \label{eq15}
J(t,x;\nu) =   E_{t,x} \left[\int_t^{T} \overline{F}(s,t)U(c(s),s) + f(s,t)B(Z(s),s)  \; \rmd s + \overline{F}(T,t)W(X(T)) \; \Bigl|\; \Fcal_t  \right]  \ ,
\end{eqnarray*}
where $\overline{F}(s,t)$ is the conditional probability for the wage earner's death to occur at time $s$ conditional upon the wage earner being alive at time $t \leq s$ and $f(s,t)$ is the corresponding conditional probability density function.
\end{lemma}

Using the previous lemma, one can state the following dynamic programming principle, obtaining a recursive relationship for the maximum expected utility as a function of the wage earner's age and his wealth at that time. See \cite{Ye} for a proof.

\begin{lemma}[Dynamic programming principle]\label{DPP}
For $0 \le t < s < T$, the maximum expected utility $V(t,x)$ satisfies the recursive relation 
\begin{eqnarray*} \label{eq16}
V(t,x) =  \underset{\nu \in \Acal(t,x)}{\sup}E\Biggl[&&\exp\left(-\int_t^s \lambda(u)\;\rmd u\right)V(s,X(s))  \nonumber \\
&&+ \int_t^s \overline{F}(u,t)U(c(u),u) + f(u,t)B(Z(u),u)  \; \rmd u \;\Bigl|\;\mathcal{F}_t\Biggr] \ .
\end{eqnarray*}
\end{lemma}

The transformation to a fixed planning horizon can then be given the following interpretation: a wage earner facing unpredictable death acts as if he will live until time $T$, but with a subjective rate of time preferences equal to his ``force of mortality'' for the consumption of his family and his terminal wealth.

\subsection{Hamilton-Jacobi-Bellman equation}

The dynamic programming principle enables us to state the HJB equation, a second-order partial differential equation whose ``solution'' is the value function of the optimal control problem under consideration here. The techniques used in the derivation of the HJB equation and the proof of the next theorem follow closely those in \cite{Fleming_Soner,Ye,Yong_Zhou}.

\begin{theorem}\label{optimal}
Suppose that the maximum expected utility $V$ is of class $C^2$. Then $V$ satisfies the Hamilton-Jacobi-Bellman equation
\begin{equation} \label{eq17}
\begin{cases}
V_t(t,x) - \lambda(t) V(t,x) + \underset{\nu \in \Acal(t,x)}{\sup} \Hcal(t,x;\nu) = 0 \\
V(T,x) = W(x)
\end{cases} \ ,
\end{equation}
where the Hamiltonian function $\Hcal$ is given by
\begin{eqnarray*}
\Hcal(t,x;\nu)  & = & \left(i(t) - c - p + \left(r(t) + \sum_{n = 1}^N{\theta_n(\mu_n(t)-r(t))}\right)x \right)V_x(t,x) \\
& & + \frac{x^2}{2} \sum_{m=1}^{M} \left( \sum_{n=1}^{N} \theta_n  \sigma_{nm}(t)\right)^2 V_{xx}(t,x)
    + \lambda(t)B\left( x + \frac{p}{\eta(t)},t\right) + U(c,t) \ .
\end{eqnarray*}
Moreover, an admissible strategy $\nu^*=(c^*, p^*, \theta^*)$ whose corresponding wealth is $X^*$ is optimal if and only if for a.e. $s\in[t,T]$ and $P$-a.s. we have
\begin{eqnarray} \label{eq18}
V_t(s,X^*(s)) - \lambda(s) V(s,X^*(s)) + \Hcal(s,X^*(s);\nu^*) = 0 \ .
\end{eqnarray}
\end{theorem}

\begin{proof}
We divide the proof in two parts: we start by establishing the HJB equation \eqref{eq17} and then we will prove that the equality \eqref{eq18} holds.

Recall that the wealth process $X(t)$, $t \in [0,\textrm{min}\{T,\tau \}]$, satisfies the stochastic differential equation \eqref{eq5}.
Using Itô's lemma, we obtain that
\begin{equation}\label{dem1}
V(t+h,X(t+h)) = V(t,X(t)) + \int_t^{t+h} a(u,X(u))\; \rmd u + \sum_{m = 1}^M \int_t^{t+h} b_m(u,X(u))\; \rmd W_m(u) \ , 
\end{equation}
where the integrand functions $a$ and $b_m$, $m=1,\ldots,M$, are given by
\begin{eqnarray}\label{a_b}
a(t,X(t)) &=&   V_t(t,X(t)) + \left(i(t)-c-p + X(t)\left(r(t)+\sum_{n = 1}^N{\theta_n\left(\mu_n(t)-r(t)\right)}\right)\right) V_x(t,X(t)) \nonumber  \\
& & + \frac{1}{2} X^2(t)V_{xx}(t,X(t)) \sum_{m = 1}^M \left(\sum_{n=1}^N{\theta_n\sigma_{nm}(t)}\right)^2  \\
b_m(t,X(t))&=& X(t)V_x(t,X(t))\sum_{n=1}^N{\theta_n\sigma_{nm}(t)}  \ , \qquad m = 1, \ldots, M  \nonumber  \ .
\end{eqnarray}
Using the dynamic programming principle of Lemma \ref{DPP} and setting $s=t+h$ we get the identity
\begin{eqnarray}\label{eq_v}
V(t,x) =  \underset{\nu \in \Acal(t,x)}{\sup}E\Biggl[&&\exp\left(-\int_t^{t+h} \lambda(u)\;\rmd u\right)V(t+h,X(t+h))  \nonumber \\
&&+ \int_t^{t+h} \overline{F}(u,t)U(c(u),u) + f(u,t)B(Z(u),u)  \; \rmd u \;\Bigl|\;\mathcal{F}_t\Biggr] \ .
\end{eqnarray}
Noting that for small enough values of $h$ the following inequalities hold
\begin{eqnarray}\label{taylor}
\textrm{exp} \left(-\int_t^{t+h} \lambda(v)~\rmd v \right) \in 1 - \lambda(t)h \pm \mathcal{O}(h^2) 
\end{eqnarray}
and combining the equality \eqref{eq_v} with the inequalities \eqref{taylor} above, we obtain
\begin{eqnarray*}
0 & \in & \underset{\nu \in \Acal(t,x)}{\sup}E\Biggl[(1 - \lambda(t)h \pm \mathcal{O}(h^2))V(t+h,X(t+h)) - V(t,x) \\
&&+ \int_t^{t+h} \overline{F}(u,t)U(c(u),u) + f(u,t)B(Z(u),u)  \; \rmd u \;\Bigl|\;\mathcal{F}_t\Biggr] \ .
\end{eqnarray*}
Substituting \eqref{dem1} in the last equation we obtain
\begin{eqnarray*}
0 & \in & \underset{\nu \in \Acal(t,x)}{\sup}E\Biggl[\left(1 - \lambda(t)h \pm \mathcal{O}(h^2)\right)\bigg(V(t,X(t)) + \int_t^{t+h} a(u,X(u))\; \rmd u\bigg) \\
&& + \left(1 - \lambda(t)h \pm \mathcal{O}(h^2)\right)\sum_{m = 1}^M \int_t^{t+h} b_m(u,X(u))\; \rmd W_m(u) \\
&&- V(t,x) + \int_t^{t+h} (1-F(u,t))U(c(u),u) + \lambda(t)(1-F(u,t))B(Z(u),u)  \; \rmd u \;\Bigl|\;\mathcal{F}_t\Biggr] \ .
\end{eqnarray*}
Dividing the previous equation by $h$ and letting $h$ go to zero we obtain the equality
\begin{eqnarray*}
0 & = & \underset{\nu \in \Acal(t,x)}{\sup}\Biggl[V_t(t,x) - \lambda(t) V(t,x) +
\left(i(t) - c - p + \left(r(t)+\sum_{n = 1}^N{\theta_n(\mu_n(t)-r(t))}\right)x \right)V_x(t,x) \\
& & + \frac{x^2}{2} \sum_{m=1}^{M} \left( \sum_{n=1}^{N} \theta_n \sigma_{nm}(t)\right)^2 V_{xx}(t,x)
 + \lambda(t)B\left(Z(t),t \right) + U(c,t) \;\Bigl|\;\mathcal{F}_t\Biggr] \ .
\end{eqnarray*}
Letting $Z(t) = x + \frac{p}{\eta(t)}$, the dynamic programming equation becomes
\begin{eqnarray*}
0 & = & \underset{\nu \in \Acal(t,x)}{\sup}\Biggl[V_t(t,x) - \lambda(t) V(t,x) +
\left(i(t) - c - p + \left(r(t)+\sum_{n = 1}^N{\theta_n(\mu_n(t)-r(t))}\right)x \right)V_x(t,x) \\
& & + \frac{x^2}{2} \sum_{m=1}^{M} \left( \sum_{n=1}^{N} \theta_n \sigma_{nm}(t)\right)^2 V_{xx}(t,x)
 + \lambda(t)B\left(x + \frac{p}{\eta(t)},t \right) + U(c,t)  \;\Bigl|\;\mathcal{F}_t\Biggr] \ .
\end{eqnarray*}
Finally, noting that $V_t(t,x) - \lambda(t) V(t,x)$ does not depend on $\nu$, we obtain the HJB equation of \eqref{eq17}, thus concluding the proof of the first part of the theorem.

Regarding the second part of the theorem, we start by letting $\nu \in \Acal(t,x)$ and apply Itô's lemma to 
\begin{equation*}
\textrm{exp} \left(-\int_t^{s} \lambda(v)~\rmd v \right) V(s,X(s)) \ .
\end{equation*}
We obtain that
\begin{eqnarray*}
V(t,x) & = & \textrm{exp}\left(-\int_t^{T} \lambda(v)~\rmd v \right) W(X(T))\nonumber\\
& & -\int_t^T \textrm{exp} \left(-\int_t^{u} \lambda(v)~\rmd v \right)  \left(a(u,X(u)) - \lambda(u)V(u,X(u)) \right)\;\rmd u  \\
& & - \sum_{m=1}^{M} \int_t^{T} \textrm{exp} \left(-\int_t^{u} \lambda(v)~\rmd v \right) b_m(u,X(u)) \; \rmd W_m(u) \nonumber \ ,
\end{eqnarray*}
where $a(t,x)$ and $b_m(t,x)$, $m=1,\ldots,M$ are as given in \eqref{a_b}. 
From the previous equality, we get
\begin{eqnarray}\label{dem2.1}
V(t,x) & = & E_{t,x} \bigg[\textrm{exp} \left(-\int_t^{T} \lambda(v)~\rmd v \right) W(X(T)) \nonumber \\
& & - \int_t^T \textrm{exp} \left(-\int_t^{u} \lambda(v)~\rmd v \right)\left(a(u,X(u)) - \lambda(u)V(u,X(u)) \right) \;\rmd u \;\Bigl|\;\mathcal{F}_t\Biggr]  \ .
\end{eqnarray}
From the definition of the hazard function in \eqref{eq9}, we obtain that the conditional probability $\overline{F}(s,t)$ for the wage earner's death to occur at time $s$ conditional upon the wage earner being alive at time $t \leq s$ is given by
\begin{equation}\label{eqF1}
\overline{F}(s,t) = \textrm{exp} \left(-\int_t^{s} \lambda(v)~\rmd v \right) \ .
\end{equation}
Similarly, we obtain that the conditional probability density function $f(s,t)$ for the death to occur at time $s$ conditional upon the wage earner being alive at time $t\leq s$ is given by
\begin{equation}\label{eqF2}
f(s,t) = \lambda(s)\textrm{exp} \left(-\int_t^{s} \lambda(v)~\rmd v \right) \ .
\end{equation}
Using \eqref{eqF1} and \eqref{eqF2}, we rewrite \eqref{dem2.1} as
\begin{equation}\label{dem2.2}
V(t,x)  =  E_{t,x} \bigg[\overline{F}(T,t) W(X(T))  - \int_t^T \overline{F}(u,t)\left(a(u,X(u)) - \lambda(u)V(u,X(u)) \right) \;\rmd u \;\Bigl|\;\mathcal{F}_t\Biggr]  \ .
\end{equation}
Using Lemma \ref{lemma1}, we rearrange \eqref{dem2.2} to obtain
\begin{eqnarray}\label{dem2.3}
\lefteqn {V(t,x) =  J(t,x;\nu) } \nonumber\\
 && - E_{t,x}\left[~\int_t^T \overline{F}(u,t) \left(V_t(u,X(u)) - \lambda(u)V(u,X(u)) + \Hcal(u,X(u);\nu) \right)\;\rmd u\;\Bigl|\;\mathcal{F}_t\right] \ .
\end{eqnarray}

Consider now an optimal admissible strategy $\nu^*=(c^*, p^*, \theta^*)$ whose corresponding wealth is $X^*$. From equation \eqref{dem2.3} we have that
\begin{eqnarray}\label{dem2.5}
\lefteqn{V(t,x) = J(t,x;\nu^*)} \nonumber \\
&& - E_{t,x}\left[~\int_t^T \overline{F}(u,t) \left(V_t(u,X^*(u)) - \lambda(u)V(u,X^*(u)) + \Hcal(u,X^*(u);\nu^*) \right)\;\rmd u\;\Bigl|\;\mathcal{F}_t\right]
\end{eqnarray}
and from the HJB equation \eqref{eq17}, we have
\begin{eqnarray}\label{dem3}
V_t(u,X^*(u)) - \lambda(u)V(u,X^*(u)) + \Hcal(u,X^*(u);\nu^*) \leq 0.
\end{eqnarray}
Combining \eqref{dem2.5} and \eqref{dem3}, we obtain that $\nu^*$ is optimal if and only if the value function $V$ satisfies \eqref{eq18}, which concludes the proof.
\end{proof}

The second part of the theorem above provides a clear approach for the computation of optimal insurance, portfolio and consumption strategies. In particular, we obtain the existence of such optimal strategies under rather weak conditions on the utility functions.

\begin{corollary}
Suppose that the maximum expected utility $V$ is of class $C^2$ and that the utility functions $U$ and $B$ are strictly concave with respect to their first variable. Then the Hamiltonian function $\Hcal$ has a regular interior maximum $\nu^*=(c^*,p^*,\theta^*)\in\Acal(t,x)$.
\end{corollary}

\begin{proof}
Using the second part of theorem \eqref{optimal}, an optimal admissible strategy $\nu^*=(c^*, p^*, \theta^*)$ with wealth process $X^*$ must satisfy \eqref{eq18}. Therefore, $\nu^*$ must be such that $\Hcal$ attains its maximum value. We start by remarking that the condition to obtain the maximum for $\Hcal$ decouples into three independent conditions, as seen in the following:
\begin{eqnarray} \label{eqa}
\underset{\nu}{\textrm{sup}} \; \Hcal (t,x;\nu) & = & (r(t)x + i(t))V_x(t,x) + \underset{c}{\textrm{sup}} \;\bigg\{U(c,t) - cV_x(t,x) \bigg\} \nonumber \\
& & + \underset{p}{\textrm{sup}}\; \bigg\{\lambda(t) B\bigg(x + \frac{p}{\eta(t)},t\bigg) - pV_x(t,x) \bigg\}  \\
& & + \underset{\theta}{\textrm{sup}}\; \bigg\{\frac{x^2}{2} \sum_{m=1}^{M} \left( \sum_{n=1}^{N} \theta_n \sigma_{nm}(t)\right)^2 V_{xx}(t,x)
+ \sum_{n=1}^{N} \theta_n(\mu_n(t) - r(t))xV_x(t,x) \bigg\} \nonumber\ .
\end{eqnarray}
Therefore, it is enough to study the variation of $\Hcal$ with respect to each one of the variables $c$, $p$ and $\theta$ independently. Thus, computing the first-order conditions for a regular interior maximum of $\Hcal$ with respect to $c$, $p$ and $\theta$ we obtain, respectively, the following three conditions
\begin{eqnarray}\label{eqe}
-V_x(t,x)+ U_c(c^*,t) &=& 0 \nonumber\\
-V_x(t,x) + \frac{\lambda(t)}{\eta(t)}B_Z\bigg(x+\frac{p^*}{\eta(t)},t\bigg) &=& 0  \\
x V_x(t,x) \alpha + x^2  V_{xx}(t,x) \sigma \sigma^T \theta^* &=& 0_{\R^N} \nonumber \ ,
\end{eqnarray}
where the subscripts in $U$ and $B$ denote differentiation with respect to each function's first variable, $\alpha$ denotes the risk premium \eqref{rp}, and $0_{\R^N}$ denotes the origin of $\R^N$.
Computing the second derivative with respect to each variable (or the Hessian matrix in the case of $\theta$), we obtain
\begin{eqnarray}\label{eqb}
\Hcal_{cc}(t,x;\nu^*) &=& U_{cc}(c^*,t) \nonumber \\
\Hcal_{pp}(t,x;\nu^*) &=& \frac{\lambda(t)}{\eta^2(t)}B_{ZZ}\bigg(x+\frac{p^*}{\eta(t)},t\bigg)  \\
\Hcal_{\theta \theta}(t,x;\nu^*) &=& x^2  V_{xx}(t,x)\sigma \sigma^T \nonumber \ .
\end{eqnarray}

Note that $\Hcal_{cc}(t,x;\nu^*)$ is negative since $U$ is strictly concave on its first variable and that $\Hcal_{pp}(t,x;\nu^*)$ is negative since $\lambda(t)$ is positive for every $0\le t\le T$ and $B$ is strictly concave on its first variable. To see that $\Hcal_{\theta \theta}(t,x;\nu^*)$ is negative definite, recall that $\sigma \sigma^T$ is assumed to be non-singular and, thus, positive definite. Moreover, note that $V_{xx}(t,x)$ must be negative: if  $V_{xx}(t,x)$ is positive, then $\Hcal$ would not be bounded above and as a consequence of the HJB equation either $V_t(t,x)$ or $V(t,x)$ would have to be infinity, contradicting the smoothness assumption on $V$. Therefore, $\Hcal_{\theta \theta}$ is negative definite and $\Hcal$ has a regular interior maximum.
\end{proof}

\section{The family  of discounted CRRA utilities}\label{CRRA}

In this section we describe the special case where the wage earner has the same discounted CRRA utility functions for the consumption of his family, the size of his legacy, and the size of his terminal wealth, i.e., from now on we assume that the utility functions are given by
\begin{equation}\label{UBW}
U(c,t) = \rme^{-\rho t}\frac{c^\gamma}{\gamma} \ , \qquad B(Z,t) = \rme^{-\rho t}\frac{Z^\gamma}{\gamma} \ , \qquad W(X) = \rme^{-\rho T}\frac{X^\gamma}{\gamma} \ ,
\end{equation}
where the risk aversion parameter $\gamma$ is such that $\gamma< 1$, $\gamma\ne 0$, and the discount rate $\rho$ is positive.

\subsection{The optimal strategies}
Using the optimality criteria provided in theorem \ref{optimal}, we obtain the following optimal strategies for discounted CRRA utility functions.
\begin{proposition}\label{Prop_optimal}
Let $\xi$ denote the non-singular square matrix given by $(\sigma \sigma^T)^{-1}$. The optimal strategies in the case of discounted constant relative risk aversion utility functions are given by
\begin{eqnarray*}\label{eq29}
c^*(t,x) &=& \frac{1}{e(t)} (x + b(t)) \\
p^*(t,x) &=& \eta(t)\left( \left( D(t) - 1 \right) x + D(t)b(t) \right)  \\
\theta^*(t,x) &=& \frac{1}{x (1 - \gamma)}(x+b(t)) \xi \alpha(t) \ ,
\end{eqnarray*}
where
\begin{eqnarray*}
b(t) &=& \int_t^T i(s) \exp \left(-\int_t^s r(v) + \eta(v) \; \rmd v \right) \rmd s \\
D(t) &=& \frac{1}{e(t)}\left(\frac{\lambda(t)}{\eta(t)} \right)^{1/(1-\gamma)} \\
e(t) &=& \exp\left(-\int_t^T H(v)	\; \rmd v \right) + \int_t^T \exp\left(-\int_t^s H(v) \;\rmd v \right) K(s)\; \rmd s \\
H(t) & = & \frac{\lambda(t) + \rho}{1-\gamma}  - \gamma \frac{\Sigma(t)}{(1-\gamma)^2 } - \frac{\gamma}{1-\gamma} (r(t) + \eta(t)) \\
K(t) &=& \frac{(\lambda(t))^{1/(1-\gamma)}}{(\eta(t))^{\gamma/(1-\gamma)}} + 1 \\
\Sigma(t) &=& \alpha^T(t) \xi \alpha(t) - \frac{1}{2} \|\sigma^T \xi \alpha(t)\|^2 \ .
\end{eqnarray*}
\end{proposition}

\begin{proof}
Assume that the utility functions $U$, $B$ and $W$ are as given in \eqref{UBW}. Using the first order conditions in \eqref{eqe}, we obtain that the optimal strategies depending on the value function $V$ are given by
\begin{eqnarray}\label{eqh}
c^*(t,x) &=& \left(e^{\rho t} V_x(t,x)\right)^{-1/(1-\gamma)} \nonumber \\
p^*(t,x) &=& \eta(t)\left(\left( \frac{\eta(t) e^{\rho t} V_x(t,x) }{\lambda(t)}\right)^{-1/(1-\gamma)} - x\right) \\
\theta^*(t) &=& - \frac{V_x(t,x)}{x V_{xx}(t,x)} \xi \alpha(t) \nonumber \ .
\end{eqnarray}

We are now going to find an explicit solution for the HJB equation \eqref{eq17}. We substitute $c$, $p$ and $\theta$ in the HJB equation by the optimal strategies in \eqref{eqh} and combine similar terms to arrive at the following partial differential equation
\begin{eqnarray}\label{eqk}
\lefteqn{ V_t(t,x) - \lambda(t) V(t,x) + \left((r(t) + \eta(t))x + i(t)\right) V_x(t,x)} \nonumber \\
&& -  \Sigma(t)  \frac{(V_x(t,x))^2}{V_{xx}(t,x)}
 + \frac{1-\gamma}{\gamma} e^{-\rho t/(1-\gamma)} K(t)
(V_x(t,x))^{-\gamma/(1-\gamma)} = 0  \ ,
\end{eqnarray}
where $\Sigma(t)$ and $K(t)$ are as given in the statement of this proposition and the terminal condition is given by
\begin{equation}\label{TC}
V(T,x) = W(x) \ .
\end{equation}
We consider an ansatz of the form
\begin{eqnarray}\label{eqm}
V(t,x) = \frac{a(t)}{\gamma}(x + b(t))^\gamma \ ,
\end{eqnarray}
and substitute it in \eqref{eqk} so that $a(t)$ and $b(t)$ are determined by the differential equation
\begin{eqnarray*}\label{eqn}
\lefteqn{\frac{1}{\gamma} \frac{\textrm{d}a(t)}{\textrm{d}t} + \frac{a(t)}{x + b(t)}\frac{\textrm{d}b(t)}{\textrm{d}t} - \lambda(t) \frac{a(t)}{\gamma} + \frac{[(r(t) + \eta(t))x + i(t)] a(t)}{x + b(t)}} \nonumber\\
&& + \Sigma(t) \frac{a(t)}{1-\gamma} + \frac{1-\gamma}{\gamma} e^{-\rho t/(1-\gamma)} K(t) (a(t))^{-\gamma/(1-\gamma)} = 0  \ .
\end{eqnarray*}
Note now that the previous differential equation and the terminal condition \eqref{TC} decouples into two independent boundary value problems for $a(t)$ and $b(t)$ which are given, respectively, by
\begin{eqnarray}\label{eqp}
&&\frac{1}{\gamma} \frac{\textrm{d}a(t)}{\textrm{d}t} + \left(r(t) + \eta(t) -  \frac{\lambda(t)}{\gamma}+\frac{\Sigma(t)}{1-\gamma}\right) a(t)  
+ \frac{1-\gamma}{\gamma} e^{-\rho t/(1-\gamma)} K(t) (a(t))^{-\gamma/(1-\gamma)}  = 0 \nonumber \\
&&a(T)=e^{-\rho T} \ ,
\end{eqnarray}
and
\begin{eqnarray}\label{eqo}
&&\frac{\textrm{d}b(t)}{\textrm{d}t} - (r(t) + \eta(t)) b(t) + i(t) = 0 \nonumber \\
&&b(T) = 0 \ .
\end{eqnarray}
To find a solution for the boundary value problem \eqref{eqp}, we write $a(t)$ in the form
\begin{equation*}\label{eqr}
a(t) = e^{-\rho t} (e(t))^{1-\gamma} \ ,
\end{equation*}
obtaining a new boundary value problem for the function $e(t)$ of the form
\begin{eqnarray}\label{eqt}
&&\frac{\textrm{d}e(t)}{\textrm{d}t} - H(t) e(t) + K(t) = 0 \nonumber \\
&&e(T) = 1 \ ,
\end{eqnarray}
where $K(t)$ and $H(t)$ are as given in the statement of this proposition.
Since equation \eqref{eqt} is a linear, non-autonomous, first order ordinary differential equation, it clearly has an explicit solution of the form
\begin{eqnarray*}\label{equ}
e(t) = \textrm{exp} \left(-\int_t^T H(v)~\textrm{d}v \right) + \int_t^T \textrm{exp} \left(-\int_t^s H(v)~\textrm{d}v \right) K(s)~~\textrm{d}s \ .
\end{eqnarray*}
Therefore, we obtain that the solution of \eqref{eqp} is given by
\begin{eqnarray}\label{eqv}
a(t) = e^{-\rho t} \left(\textrm{exp} \left(-\int_t^T H(v)~\textrm{d}v \right) +
\int_t^T \textrm{exp} \left(-\int_t^s H(v)~\textrm{d}v \right)~K(s)~~\textrm{d}s \right)^{1-\gamma} \ .
\end{eqnarray}
To find a solution for the boundary value problem \eqref{eqo}, we just note that this is again a linear, non-autonomous, first order differential equation and its solution is given by
\begin{eqnarray}\label{eqq}
b(t) = \int_t^T i(s)~\textrm{exp} \left(-\int_t^s r(v) + \eta(v)~\textrm{d}v \right)~\textrm{d}s 
\end{eqnarray}
as required.

Combining \eqref{eqh} with \eqref{eqm}, \eqref{eqv} and \eqref{eqq}, we obtain that the optimal strategies in the case of CRRA utilities are then given by
\begin{eqnarray*}\label{eqw}
c^*(t,x) &=& \frac{1}{e(t)} (x + b(t)) \nonumber \\
p^*(t,x) &=& \eta(t)((D(t)-1)x + D(t)b(t)) \\
\theta^*(t,x) &=& \frac{x+b(t)}{x (1 - \gamma)}\xi \alpha(t) \nonumber  \ ,
\end{eqnarray*}
where $D(t)$ is as given in the statement of this proposition, which concludes the proof.
\end{proof}

Note that the quantities $b(t)$ and $x+b(t)$ are of essential relevance for the definition of the optimal strategies in proposition \ref{Prop_optimal}. The quantity $b(t)$, that we will refer to as \emph{human capital} following the nomenclature introduced in \cite{Pliska_Ye_1}, should be seen as representing the fair value at time $t$ of the wage earner's future income from time $t$ to time $T$, while the quantity $x+b(t)$ should be thought of as the full wealth (present wealth plus future income) of  the wage earner at time $t$. It is then natural that these two quantities play a central role in the choice of optimal strategies, since they determine the present and future wealth available for the wage earner and his family.

\begin{remark}
Noting that $r(t)$, $\eta(t)$ and $i(t)$ are positive functions and considering the boundary value problem \eqref{eqq}, if $r(t)+\eta(t)$ is small enough, we can deduce that human capital function $b(t)$ has the following properties:
\begin{itemize}
\item[a)] it is a positive function for all $0\le t< T$;
\item[b)] it is concave.
\end{itemize}
Moreover, we have that $b(t)$ is either:
\begin{itemize}
\item[i)] a decreasing function for all $t\in[0,T]$; or
\item[ii)] a unimodal map of $t$, i.e. there exists some $t^*\in(0,T)$ such that $b(t)$ is increasing for all $0<t<t^*$, decreasing for all $t^*<t<T$; Furthermore, we have that the graph of $b(t)$ intersects the graph of the function $i(t)/(r(t)+\eta(t))$ at $t=t^*$. 
\end{itemize}
\end{remark}

From the explicit knowledge of the optimal strategies, several economically relevant conclusions can be obtained. See Fig. \ref{Poptimal_fig} for a graphical representation of the optimal life-insurance purchase as a function of age and ``full wealth'' $x+b(t)$ of the wage earner. We start by proving an auxiliary lemma before moving on to the statement and proof of some of the optimal strategies properties.

\begin{figure}[h!]
	\centering
      \psfrag{p}[cc][][0.85][0]{$p^*(t,x)$}%
      \psfrag{x+b}[cc][][0.85][0]{$x+b(t)$}
      \psfrag{t}[cc][][0.85][0]{$t$}
      \psfrag{0}[cc][][0.85][0]{$0$}
      \psfrag{1000}[cc][][0.85][0]{$1000$}
      \psfrag{2000}[cc][][0.85][0]{$2000$}
      \psfrag{3000}[cc][][0.85][0]{$3000$}
      \psfrag{30}[cc][][0.85][0]{$30$}
      \psfrag{40}[cc][][0.85][0]{$40$}
      \psfrag{50}[cc][][0.85][0]{$50$}
      \psfrag{60}[cc][][0.85][0]{$60$}
      \psfrag{5}[cc][][0.85][0]{$5$}
      \psfrag{10}[cc][][0.85][0]{$10$}
      \psfrag{-5}[cc][][0.85][0]{$-5$}
      \psfrag{-10}[cc][][0.85][0]{$-10$}
%%%% \psfrag{texto a substituir}{rcl e tcb ou contrario}[][escala das letras][angulo para rodar letras]{texto a incluir}
   		\includegraphics[width=90mm,angle=-90]{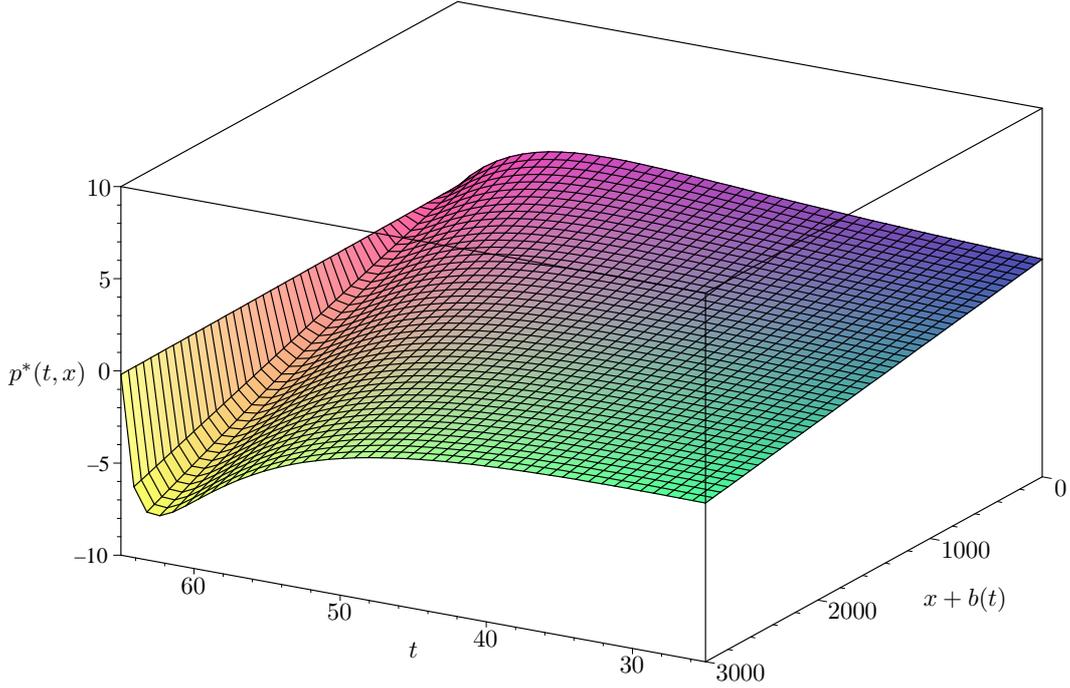}
   		\caption{The optimal life-insurance purchase for a wage earner that starts working at age 25 and retires 40 years later. The parameters of the model were taken as $N=M=2$, $i(t)=50000\exp(0.03t)$, $r=0.04$, $\rho=0.03$, $\gamma=-3$, $\lambda(t)=0.001+\exp(-9.5+0.1t)$, $\eta(t)=1.05\lambda(t)$, $\mu_1=0.07$, $\mu_2=0.11$, $\sigma_{11}=0.19$, $\sigma_{12}=0.15$, $\sigma_{21}=0.17$ and $\sigma_{22}=0.21$.}\label{Poptimal_fig}
\end{figure}

\begin{lemma}\label{lemmaD}
Suppose that for all $t\in[0,\min\{T,\tau\}]$ the following two conditions are satisfied:
\begin{itemize}
\item[a)] $\lambda(t)\le \eta(t)$;
\item[b)] $H(t)\le 1$.
\end{itemize}
Then, the inequality $D(t) < 1$ holds for every $t\in[0,\min\{T,\tau\}]$.
\end{lemma}
\begin{proof}
Recall that $D(t)$ is given by
\begin{equation*}
D(t) = \frac{1}{e(t)}\left(\frac{\lambda(t)}{\eta(t)} \right)^{1/(1-\gamma)} \ .
\end{equation*}
Using condition b) and noting that $K(t)$ is positive for all $t\in[0,\min\{T,\tau\}]$, we have that
\begin{eqnarray*}
e(t) &=& \textrm{exp} \left(-\int_t^T H(v)\;\textrm{d}v \right) + \int_t^T \textrm{exp} \left(-\int_t^s H(v)\;\textrm{d}v \right) K(s)\;\textrm{d}s \\
     &\geq & \textrm{exp} \left(-\int_t^T 1\;\textrm{d}v \right) + \int_t^T \textrm{exp} \left(-\int_t^s 1\;\textrm{d}v \right) K(s)\;\textrm{d}s > 1 \ .
\end{eqnarray*}
Putting together the previous inequality and condition a), we obtain the required inequality.
\end{proof}

The next result provides a qualitative characterization of the optimal life insurance purchase strategy.
\begin{corollary}\label{cor_p_x_t}
Assume that the conditions of lemma \ref{lemmaD} are satisfied. Then, the optimal insurance purchase strategy $p^*(t,x)$ has the following properties:
\begin{itemize}
\item[a)] it is a decreasing function of the wealth $x$;
\item[b)] it is an increasing function of the wage earner's human capital $b(t)$;
\item[c)] it is negative for suitable pairs of wealth $x$ and ``age'' $t$;
\item[d)] if the wage earner's wealth $x$ is small enough and $\eta(t)$ is non-decreasing, the function $ t\mapsto p^*(t,x-b(t))$ has the same monotonicity as the human capital function $b(t)$. 
\end{itemize}
\end{corollary}
\begin{proof}
Recall from proposition \ref{Prop_optimal} that the optimal insurance purchase strategy $p^*(t,x)$ is given by
\begin{eqnarray}\label{cor_p_x_t_eq1}
p^*(t,x) = \eta(t)\left( \left( D(t) - 1 \right) x + D(t)b(t) \right) 
\end{eqnarray}
for all $t\in[0,\min\{T,\tau\}]$. 

Items a) and b) follow from lemma \ref{lemmaD}, since $D(t)$ is a positive function such that $D(t)<1$ for all $t\in[0,\min\{T,\tau\}]$.

For the proof of item c), note that $p^*(t,x)$ is negative for all $(t,x)\in [0,T]\times\R^+$ such that
\begin{eqnarray*}
x &>& \frac{D(t)}{1-D(t)}b(t)  \\
 &=& \frac{\lambda(t)^{1/(1-\gamma)}}{e(t)\eta(t)^{1/(1-\gamma)} - \lambda(t)^{1/(1-\gamma)}} b(t) > 0 \ ,
\end{eqnarray*}
and positive otherwise.

Item d) follows from \eqref{cor_p_x_t_eq1} and the fact that $\eta(t)$ is a non-decreasing.
\end{proof}

Some comments regarding the assumptions in lemma \ref{lemmaD} (and corollary \ref{cor_p_x_t}) seem necessary. Starting with condition (a), the life insurance company must establish the premium-insurance $\eta(t)$ in such a way that $\lambda(t)\le \eta(t)$ in order to make a profit (the insurance policy being fair whenever $\lambda(t)\le \eta(t)$). Regarding condition (b), we note that the quantities $r$, $\rho$, $\eta$ and $\lambda$ are usually very small in the real world and, moreover, the relative risk aversion of the wage earner is negative in general. This is consistent with the assumption that $H(t)$ is bounded above by some positive constant.

Apart from studying how optimal life insurance purchase varies with age and wealth, it is also relevant to understand how the remaining parameters which define the financial and insurance markets influence life insurance purchase.

\begin{corollary}
With all other parameters, including $t$ and $x$ constant, the optimal life insurance purchase rate $p^*(t,x)$ is an increasing function of the discount rate $\rho$.
\end{corollary}
\begin{proof}
To study the influence of the discount rate $\rho$ on the optimal optimal life insurance purchase, we consider two different values of $\rho$ and compare the corresponding values of the optimal life insurance purchase. We distinguish the functions associated with each of the two parameter values by their subscript.  

Assume that $\rho_1$ and $\rho_2$ are such that $\rho_1 < \rho_2$. Recall the definitions of $p^*(t,x)$, $H(t)$, $e(t)$ and $D(t)$ given in proposition \ref{Prop_optimal}. Then, it is clear that the inequalities
\begin{eqnarray*}
H_1(t) &<& H_2(t) \nonumber \\
e_1(t) &>& e_2(t) \nonumber \\
D_1(t) &<& D_2(t)  \nonumber 
\end{eqnarray*}
hold for all $t\in[0,\min\{T,\tau\}]$, where the subscripts correspond to $\rho_1$ and $\rho_2$ in an obvious manner. Rewriting $p^*(t,x)$ as 
\begin{equation*}
p^*(t,x) = \eta(t)(D(t)(x+b(t))-x) \ ,
\end{equation*}
it follows from the preceding inequalities that
\begin{equation*}
p_1^*(t,x) < p_2^*(t,x)
\end{equation*}
for all $t\in[0,\min\{T,\tau\}]$, concluding the proof of the statement.
\end{proof}

\begin{remark}
The variation of the optimal life insurance purchase rate $p^*(t,x)$ with respect to the interest rate $r(t)$, the risk aversion parameter $\gamma$, the hazard rate $\lambda(t)$ and the insurance premium-payout ratio $\eta(t)$ is non-trivial. However, by studying the function $p^*(t,x)$ given in proposition \ref{Prop_optimal} we can make the following observations:
\begin{itemize}
\item[i)] $p^*(t,x)$ is a decreasing function of the interest rate $r(t)$, except for large values of $x$ and $t$ close enough to $T$;
\item[ii)] $p^*(t,x)$ is a decreasing function of the risk aversion parameter $\gamma$, except for values of $t$ close enough to $T$;
\item[iii)] $p^*(t,x)$ is an increasing function of the hazard rate $\lambda(t)$ and the insurance premium-payout ratio $\eta(t)$ for small enough values of wealth $x$ and a decreasing function for large values of $x$.
\end{itemize}
\end{remark}

The extra risky securities in our model introduce novel features to the wage earner's portfolio management, as is exemplified in the following result.

\begin{corollary}\label{prop_opt_thetas}
Let $\xi$ denote the non-singular square matrix given by $(\sigma \sigma^T)^{-1}$ and let $(\xi \alpha(t))_n$ denote the $n$-th component of the vector $\xi \alpha(t)$. The optimal portfolio process $\theta^*(t,x)=\left(\theta_1^*,...,\theta_N^*\right)$ is such that for every $n\in\{1,...,N\}$:
\begin{itemize}
\item[a)] $\theta_n^*$ has the same sign as $(\xi \alpha(t))_n$;
\item[b)] $\theta_n^*$ is a decreasing function of the total wealth $x$ if $(\xi \alpha(t))_n>0$ for all $t\in[0,\min\{T,\tau\}]$ and an increasing function of $x$ if $(\xi \alpha(t))_n<0$ for all $t\in[0,\min\{T,\tau\}]$;
\item[c)] $\theta_n^*$ is an increasing function of the wage earner's human capital $b(t)$ if $(\xi \alpha(t))_n>0$ for all $t\in[0,\min\{T,\tau\}]$ and a decreasing function of $b(t)$ if $(\xi \alpha(t))_n<0$ for all $t\in[0,\min\{T,\tau\}]$.
\end{itemize}
Furthermore, for every $n,m\in\{1,...,N\}$ the following equalities hold
\begin{eqnarray*}
\underset{x\rightarrow 0^+}{\lim} \theta_n^*(t,x) &=& +\infty\hspace{1cm}   \qquad \underset{x\rightarrow 0^+}{\lim} \frac{\theta_n^*(t,x)}{\theta_m^*(t,x)} =  \frac{(\xi \alpha(t))_n}{(\xi \alpha(t))_m} \\
\underset{x\rightarrow \infty}{\lim} \theta_n^*(t,x) &=& \frac{(\xi \alpha(t))_n}{1 - \gamma} \qquad \hspace{0.3cm}
\underset{t\rightarrow T}{\lim} \theta_n^*(t,x) = \frac{(\xi \alpha(T))_n}{1 - \gamma}  \ .
\end{eqnarray*}
\end{corollary}
\begin{proof}
Items a), b) and c) on the first part of the corollary follow from the form of $\theta_n^*$, $n\in\{1,...,N\}$, given in the statement of proposition \ref{Prop_optimal} and positivity of $b(t)$.

The limiting behaviours on the second part of the corollary also follow from the form of $\theta_n^*$, $n\in\{1,...,N\}$.
\end{proof}

\begin{remark}
Corollary \ref{prop_opt_thetas} is a mutual fund result: the relative proportions among the risky securities are independent of all parameters except for the interest rate and the risky assets appreciation rates and volatilities since, for any $n,m\in\{1,...,N\}$ we have 
\begin{equation*}
\frac{\theta_n^*(t,x)}{\theta_m^*(t,x)} = \frac{(\xi \alpha(t))_n}{(\xi \alpha(t))_m} \ .
\end{equation*}
\end{remark}

The quantities $(\xi \alpha(t))_n$, $n\in\{1,...,N\}$ can be though as ``weighted risk premiums'' for the risky assets $S_1,\ldots,S_N$, where the weights are provided by (quadratic) functions on the coefficients of the matrix of risky-assets volatilities $\sigma$.  

Under the assumption that the ``weighted risk premiums'' $(\xi \alpha(t))_n$ are positive for all $n\in\{1,...,N\}$ and $t\in[0,\min\{T,\tau\}]$, we obtain the following interesting consequence of the previous corollary. 
\begin{corollary}
Let $\xi$ denote the non-singular square matrix given by $(\sigma \sigma^T)^{-1}$ and let $(\xi \alpha(t))_n$ denote the $n$-th component of the vector $\xi \alpha(t)$. Assume that for every $n\in\{1,...,N\}$ we have $(\xi \alpha(t))_n>0$ for all $t\in[0,\min\{T,\tau\}]$. Then the optimal strategy for wage earners with small enough wealth $x$ is to short the risk-free security and hold higher amounts of risky assets.
\end{corollary}
\begin{proof}
The corollary follows from corollary \ref{prop_opt_thetas} and the fact that the wage earner has no budget limitations, thus allowing him to get into short positions on the risk free asset $S_0$ of arbitrary size.  
\end{proof}

We conclude this session with a result concerning some qualitative properties of the optimal consumption strategy. Its proof follows trivially from the form of the optimal consumption $c^*(t,x)$ given in proposition \ref{Prop_optimal}.
\begin{corollary}
The optimal consumption rate $c^*(t)$ is an increasing function of both the wealth $x$ and the human capital $b(t)$.
\end{corollary}

\subsection{The interaction between life insurance purchase and portfolio management}

In this section we compare the optimal life-insurance strategies for a wage earner who faces the following two situations:
\begin{itemize}
\item[a)] in the first case, we assume that the wage earner has access to an insurance market as described above and that his goal is to maximize the combined utility of his family consumption for all $t \leq \min\{T,\tau\}$, his wealth at retirement date $T$ if he lives that long, and the value of his estate in the event of premature death. The optimal strategies for the wage earner in this setting are given in Proposition \ref{Prop_optimal}.
\item[b)] in the second case, we assume that the wage earner is without the opportunity of buying life insurance. His goal is to maximize the combined utility of his family consumption for all $t \leq \min\{T,\tau\}$ and his wealth at retirement date $T$ if he lives that long. Similarly to what we have done previously, we translate this situation to the language of stochastic optimal control and derive explicit solutions in the case of discounted CRRA utilities.
\end{itemize}

We concentrate on the case b) described above for the moment. Similarly to what was done in case a), this problem can be formulated by means of optimal control theory. The wage earner's goal is then to maximize a new cost functional subject to:
\begin{itemize}
\item the dynamics of the state variable, i.e., the dynamics of a wealth process $X^0(t)$ given by
\begin{equation*} \label{eq4_0}
X^0(t) = x + \int_0^t i(s)-c^0(s)\;\rmd s+\sum_{n = 0}^N\int_0^t\frac{\theta_n^0(s)X^0(s)}{S_n(s)}\rmd S_n(s) \ ,
\end{equation*}
where $t \in [0, \textrm{min}\{T,\tau\}]$ and $x$ is the wage earner's initial wealth.
\item constraints on the remaining control variables, i.e., the consumption process $c^0(t)$ and the reduced portfolio process $\theta^0(t)=\left(\theta_1^0(t),\cdots,\theta_N^0(t)\right)\in\R^{N}$; and
\item boundary conditions on the state variables.
\end{itemize}

Let us denote by $\Acal^0(x)$ the set of all admissible decision strategies, i.e. all admissible choices for the control variables $\nu^0 = (c^0,\theta^0)\in\R^{N+1}$. The dependence of $\Acal^0(x)$ on $x$ denotes the restriction imposed on the wealth process by the boundary condition $X^0(0)=x$.

The wage earner's problem can then be stated as follows: find a strategy $\nu^0=(c^0,\theta^0)\in\Acal^0(x)$ which maximizes the expected utility
\begin{eqnarray} \label{eq7_0}
V^0(x) = \underset{\nu^0 \in \Acal^0(x)}{\sup}\; E_{0,x} \Biggl[\int_0^{T \wedge \tau}U(c^0(s),s) \; \rmd s + W(X^0(T))I_{\{\tau>T\}} \;   \Biggr] \ ,
\end{eqnarray}
where $U(c^0,\cdot)$ is again the utility function describing the wage earner's family preferences regarding consumption in the time interval $[0,\min\{T,\tau\}]$ and $B(X^0,t)$ is the utility function for the terminal wealth at time $t=T\wedge\tau$. As before, we restrict ourselves to  the special case where the wage earner has the same discounted CRRA utility functions for the consumption of his family and his terminal wealth given in \eqref{UBW}. The optimal strategies are given in the next result

\begin{proposition}\label{Prop_optimal_0}
Let $\xi$ denote the non-singular square matrix given by $(\sigma \sigma^T)^{-1}$. The optimal strategies for problem \eqref{eq7_0} in the case where $U(c^0,\cdot)$ and $B(X^0,\cdot)$ are the discounted constant relative risk aversion utility functions in \eqref{UBW} are given by
\begin{eqnarray*}\label{eq29_0}
{c^0}^*(t,x) &=& \frac{1}{e^0(t)} (x + b^0(t)) \\
{\theta^0}^*(t,x) &=& \frac{1}{x (1 - \gamma)}(x+b^0(t)) \xi \alpha(t) \ ,
\end{eqnarray*}
where
\begin{eqnarray*}
b^0(t) &=& \int_t^T i(s) \exp \left(-\int_t^s r(v) \; \rmd v \right) \rmd s \\
e^0(t) &=& \exp\left(-\int_t^T H^0(v)	\; \rmd v \right) + \int_t^T \exp\left(-\int_t^s H^0(v) \;\rmd v \right)\; \rmd s \\
H^0(t) & = & \frac{\lambda(t) + \rho}{1-\gamma}  - \gamma \frac{\Sigma(t)}{(1-\gamma)^2 } - \frac{\gamma}{1-\gamma}r(t) \\
\end{eqnarray*}
and $\Sigma(t)$ is as given in the statement of proposition \ref{Prop_optimal}.
\end{proposition}

We skip the proof of the previous proposition, since it is of the same nature as the proofs of theorem \ref{optimal} and proposition \ref{Prop_optimal}. 

We should point that it would be preferable from an economic standpoint to maximize also the utility of wealth at the time of premature death. In fact, it is easy to check that this corresponds to the addition of a constraint of the form $p(t,x)=0$ to our original problem. However, the corresponding HJB equation would contain an extra term of the form $\lambda(t)\; x^\gamma /\gamma$. The presence of this additional term in the HJB equation makes it impossible for us to obtain closed form solutions. On the other hand, a strategy optimizing the final wealth at retirement should be close to an optimal strategy which maximizes also the wealth for the case of an eventual death of the wage earner before retirement time. We plan to address such a comparison in future research. 

If we make the wage earner income $i(t)$ equal to zero, then we get a solution which is close to Merton's classical solution, but it still depends on the hazard function $\lambda(t)$, i.e., even in the absence of a life insurance policy the uncertainty regarding wage earner's lifetime still plays a role on the determination of the optimal consumption and investment strategies.

Propositions \ref{Prop_optimal} and \ref{Prop_optimal_0} provide us with optimal portfolio processes for the two settings a) and b) described above. In the next theorem we show how these optimal portfolio processes compare. 

\begin{theorem}
Let $\xi$ denote the non-singular square matrix given by $(\sigma \sigma^T)^{-1}$ and $(\xi \alpha(t))_n$ the $n$-th component of the vector $\xi \alpha(t)$. For each $n\in\{1,...,N\}$, we have that ${\theta_n^0}^*(t,x)>\theta_n^*(t,x)$ if and only if $(\xi \alpha(t))_n>0$.
\end{theorem}

\begin{proof}
Recall the definitions of ${\theta^0}^*$ and $\theta^*$ from propositions \ref{Prop_optimal_0} and \ref{Prop_optimal}, respectively.
Note that
\begin{eqnarray}\label{theta_compar}
{\theta_n^0}^*(t,x) - \theta_n^*(t,x) & = & \frac{1}{x (1 - \gamma)}(x+b^0(t)) (\xi \alpha(t))_n - \frac{1}{x (1 - \gamma)}(x+b(t)) (\xi \alpha(t))_n \nonumber \\
& = & \frac{(\xi \alpha(t))_n}{x (1 - \gamma)}  (b^0(t)-b(t))  \ .
\end{eqnarray}
Using the definitions of $b^0(t)$ and $b(t)$ given in the statements of propositions \ref{Prop_optimal_0} and \ref{Prop_optimal}, respectively, we obtain that 
\begin{equation}\label{beta_compar}
b^0(t)-b(t) = \int_t^T i(s) \textrm{exp} \left(-\int_t^s r(v)\;\textrm{d}v \right) \left(1-\textrm{exp}
\left(-\int_t^s \eta(v)~\textrm{d}v \right) \right)\;\textrm{d}s \ .
\end{equation}
Since $\eta(t)$ is a positive function, we get that
\begin{equation*}
\int_t^s \eta(v)\;\textrm{d}v > 0
\end{equation*}
and therefore, the inequality
\begin{equation}\label{eta_ineq}
1-\textrm{exp} \left(-\int_t^s \eta(v)\;\textrm{d}v \right) > 0
\end{equation}
holds for all $0\le t\le s\le T$. 
Therefore, combining \eqref{theta_compar}, \eqref{beta_compar} and \eqref{eta_ineq} and recalling that $\gamma<1$, we obtain that the sign of ${\theta_n^0}^*(t,x) - \theta_n^*(t,x)$ is the same as the sign of $(\xi \alpha(t))_n$ for all $t\in[0,\min\{T,\tau\}]$.
\end{proof}

The economic implications of the theorem above are made clear in the following result.

\begin{corollary}\label{cons_cor}
Let $\xi$ denote the non-singular square matrix given by $(\sigma \sigma^T)^{-1}$ and $(\xi \alpha(t))_n$ the $n$-th component of the vector $\xi \alpha(t)$. Assume that for every $n\in\{1,...,N\}$ we have that $(\xi \alpha(t))_n>0$.
Then, the optimal portfolio of a wage earner with the possibility of buying a life insurance policy is more conservative than the optimal portfolio of the same wage earner if he does not have the opportunity to buy life insurance.
\end{corollary}

\section{Conclusions} \label{conclusion}

We have introduced a model for optimal insurance purchase, consumption and investment for a wage earner with an uncertain lifetime with an underlying financial market consisting of one risk-free security and a fixed number of risky securities with diffusive terms driven by multidimensional Brownian motion.

When we restrict ourselves to the case where the wage earner has the same discounted CRRA utility functions for the consumption of his family, the size of his legacy and his terminal wealth, we obtain explicit optimal strategies and describe new properties of these optimal strategies. Namely, we obtain economically relevant conclusions such as: (i) a young wage earner with smaller wealth has an optimal portfolio with larger values of volatility and higher expected returns; and (ii) a wage earner who can buy life insurance policies will choose a more conservative portfolio than a similar wage earner who is without the opportunity to buy life insurance.

It is worth noting that the model described in this paper can be improved by adding further ingredients such as, for instance, some form of stochasticity on the income function or the hazard rate, some constraints on the possibility of trading of stocks on margin, the existence of additional life insurance products on the market or even adding some correlation between the wage earner's mortality rate and the underlying financial market. We plan to address at least some of these issues in a future publication.

\section*{Acknowledgments}

We thank the Calouste Gulbenkian Foundation, PRODYN-ESF, POCTI, and POSI by FCT and Minist\'erio da Ci\^encia, Tecnologia e Ensino Superior, CEMAPRE, LIAAD-INESC Porto LA, Centro de Matem\'atica da Universidade do Minho and Centro de Matem\'atica da Universidade do Porto for their financial support. D. Pinheiro's research was supported by FCT - Funda\c{c}\~ao para a Ci\^encia e Tecnologia program ``Ci\^encia 2007''. I. Duarte's research was supported by FCT - Funda\c{c}\~ao para a Ci\^encia e Tecnologia grant with reference SFRH / BD / 33502 / 2008.

\bibliography{biblio}
\bibliographystyle{plain}

\end{document}